\documentclass[conference]{IEEEtran}

\usepackage[utf8]{inputenc}
\usepackage[T1]{fontenc}
\usepackage{lmodern}
\usepackage[hidelinks]{hyperref}
\usepackage{microtype}
\usepackage{amsmath}
\usepackage{amsthm}
\usepackage{amssymb}
\usepackage{booktabs}
\usepackage{balance}
\usepackage{enumerate}
\usepackage{array}
\usepackage{amsfonts}
\usepackage{mathtools}
\usepackage{balance}

\usepackage{tikz}
\usetikzlibrary{automata,arrows}
\usetikzlibrary{positioning}
\tikzset{font=\footnotesize}
\tikzstyle{every state}=[inner sep=4pt, minimum size=8pt]

\newcommand\Z{\mathbb{Z}}
\newcommand\N{\mathbb{N}}
\newcommand\R{\mathbb{R}}

\newcommand\D{\mathbb{D}}
\newcommand\B{\mathbb{B}}

\newcommand{\pref}[1]{\mathit{pref}(#1)}

\newcommand{\req}{{\it req}}
\newcommand{\ack}{{\it ack}}
\newcommand{\other}{{\it other}}

\newcommand{\tr}{\mathtt{T}}
\newcommand{\fa}{\mathtt{F}}

\newcommand{\limavg}{\mathsf{lim\,avg}}
\newcommand{\art}{\mathsf{art}}
\newcommand{\mrt}{\mathsf{mrt}}
\newcommand{\overbar}[1]{\mkern 1.5mu\overline{\mkern-1.5mu#1\mkern-1.5mu}\mkern 1.5mu}

\newtheorem{remark}{Remark}
\newtheorem{definition}{Definition}
\newtheorem{example}{Example}
\newtheorem{theorem}{Theorem}
\newtheorem{proposition}{Proposition}
\newtheorem{corollary}{Corollary}

\begin{document}
\title{Quantitative and Approximate Monitoring}

\author{\IEEEauthorblockN{Thomas A. Henzinger}
\IEEEauthorblockA{IST Austria\\
tah@ist.ac.at}
\and
\IEEEauthorblockN{N. Ege Sara\c{c}}
\IEEEauthorblockA{IST Austria\\
ege.sarac@ist.ac.at}}
	
\maketitle

\begin{abstract}
	In runtime verification, a monitor watches a trace of a system and, if possible, decides after observing each finite prefix whether or not the unknown infinite trace satisfies a given specification.
	We generalize the theory of runtime verification to monitors that attempt to estimate numerical values of quantitative trace properties (instead of attempting to conclude boolean values of trace specifications), such as maximal or average response time along a trace.
	Quantitative monitors are approximate: with every finite prefix, they can improve their estimate of the infinite trace's unknown property value.
	Consequently, quantitative monitors can be compared with regard to a precision-cost trade-off: better approximations of the property value require more monitor resources, such as states (in the case of finite-state monitors) or registers, and additional resources yield better approximations.
	We introduce a formal framework for quantitative and approximate monitoring, show how it conservatively generalizes the classical boolean setting for monitoring, and give several precision-cost trade-offs for monitors.
	For example, we prove that there are quantitative properties for which every additional register improves monitoring precision.
\end{abstract}


%

\section{Introduction} \label{sect:intro}

We provide a theoretical framework for the convergence of two recent trends in computer-aided verification.
The first trend is {\it runtime verification} \cite{Bartocci18}.
Classical verification aspires to provide a judgment about all possible runs of a system;
runtime verification, or {\it monitoring}, provides a judgment about a single, given run.
There is a trend towards monitoring because the classical ``verification gap'' keeps widening:
while verification capabilities are increasing, system complexity is increasing more quickly,
especially in the time of many-core processors, cloud computing, cyber-physical systems, and neural networks.
Theoretically speaking,
the paradigmatic classical verification problem is {\it emptiness} of the product between system and negated specification
(“does some run of the given system violate the given specification?”),
whereas the central runtime verification problem is {\it membership}
(“does a given run satisfy a given specification?”).
Since membership is easier to solve than emptiness, this has ramifications for specification formalisms;
in particular, there is no need to restrict ourselves to $\omega$-regular specifications or finite-state monitors.
We do restrict ourselves to the {\it online} setting,
where a monitor watches the finite prefixes of an infinite run and, with each prefix, renders a {\it verdict},
which could signal a satisfaction or violation of the specification, or “don’t know yet.”

The second trend is {\it quantitative verification}
\cite{Kwiatkowska07,Henzinger13}.
While classical verification is boolean, in that every complete run either satisfies or violates the specification and,
accordingly, the system is either correct (i.e., without a violating run) or incorrect,
quantitative verification provides additional, often numerical information about runs and systems.
For example, a quantitative specification may measure the probability of an event,
the ``response time'' or the use of some other resource along a run,
or by how much a run deviates from a correct run.
In {\it quantitative runtime verification}, we wish to observe, for instance,
the maximal or average response time along a given run, not across all possible runs.
Quantitative verification is interesting for an important reason beyond its ability to provide non-boolean information:
it may provide {\it approximate} results \cite{Boker14}.
A monitor that under- or over-approximates a quantitative property may be able to do so with fewer computational
resources than a monitor that computes a quantitative property's exact value.
We provide a theoretical framework for {\it quantitative and approximate monitoring},
which allows us to formulate and prove such statements.

In boolean runtime verification frameworks, there are several different notions of {\it monitorability}
\cite{Kim02,Falcone12,Aceto19a}.
Along with safety and co-safety, a well-studied definition is by \cite{Pnueli06b} and \cite{Bauer11}:
after watching any finite prefix of a run, if a positive or negative verdict has not been reached already,
there exists at least one continuation of the run which will allow such a verdict.
This {\em existential} definition is popular because a {\em universal} definition,
that on every run a positive or negative verdict will be reached eventually, is very restrictive;
only boolean properties that are both safe and co-safe can be monitored universally \cite{Aceto19a}.
By contrast, the existential definition covers finite boolean combinations of safety and co-safety, and more \cite{Falcone12}.
In \emph{quantitative approximate} monitoring, however, there is less need to prefer an existential definition of monitoring
because usually many approximations are available, even if some are poor.
The main attention must shift, rather, to the quality---i.e., precision---of the approximation.
Our quantitative framework fully 
generalizes the standard boolean versions of monitorability in a universal setting where
monitors yield approximate results on all runs and can be compared regarding their precision and resource use.
In fact, we advocate the consideration of precision-resource trade-offs as a central design criterion for monitors,
which requires a formalization of monitoring in which precision-resource trade-offs can be analyzed.
\emph{Such a formalization is the main contribution of this paper.}

As an example,
let us illustrate a precision-resource trade-off that occurs when using \emph{register machines} as monitors.
Consider a server that processes requests.
Each trace of the server is an infinite word over the alphabet $\{\req,\ack,\other\}$ of events.
An interesting quantitative property of the server is \emph{maximal response time},
which measures the maximal number of events before each $\req$ event in a trace is followed by an $\ack$ event. 
This property, denoted $p_1$, is a function that maps every infinite word to a value in $\N \cup \{\infty\}$.
To construct a precise online monitor for $p_1$,
we need two counter registers $x$ and $y$ and the ability to compare their values:
as long as $x < y$, register $x$ counts the current response time,
and $y$ stores the maximal response time encountered so far;
if $x=y$, counting continues in $y$, and $x$ is reset to 0.
The output, or \emph{verdict} value, of the monitor is always $y$.
In this way the 2-counter monitor $M_{\textit{max}}$ generates the verdict function depicted in Figure~\ref{fig:mrtVerdict}.

\begin{figure}[h]
	\centering
	\scalebox{0.922}{
		\begin{tikzpicture}
			\draw[->] (0,0) -- (6.5,0) node[anchor=west] {$f$};
			\draw
			(0.75,0) node[anchor=north] {$\req\vphantom{k}$}
			(1.5,0) node[anchor=north] {$\ack\vphantom{k}$}
			(2.25,0) node[anchor=north] {$\req\vphantom{k}$}
			(3,0) node[anchor=north] {$\other\vphantom{k}$}
			(3.75,0) node[anchor=north] {$\ack\vphantom{k}$}
			(4.5,0) node[anchor=north] {$\req\vphantom{k}$}
			(5.25,0) node[anchor=north] {$\ack\vphantom{k}$}
			(6,0) node[anchor=north] {$\other\vphantom{k}$};
			\draw[->] (0,0) -- (0,2.5) node[anchor=south] {$M_{\textit{max}}$};
			\draw
			(0,1) node[anchor=east] {1}
			(0,2) node[anchor=east] {2};	
			\draw[thick] (0,0) -- (0.75,0) -- (1.5,1) -- (2.25,1) -- (3,1) -- (3.75,2) -- (4.5,2) -- (5.25,2) -- (6,2);
		\end{tikzpicture}
	}
	\caption{Monitoring maximal response time for a trace $f$.}
	\label{fig:mrtVerdict}
\end{figure}
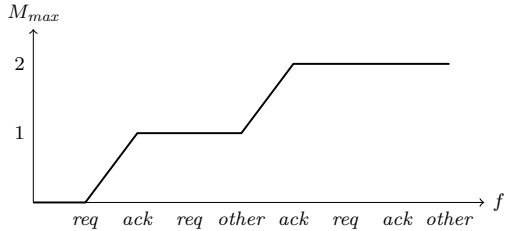

Considering the same server, one may also be interested in the \emph{average response time} of a trace.
The precise monitoring of average response time requires 3 counters and division between counter registers to generate
outputs.
Moreover, verdict values can fluctuate along a trace, producing a non-monotonic verdict function.
Figure~\ref{fig:artVerdict} shows the verdict function generated by a 3-register monitor $M_{\textit{avg}}$ with division.

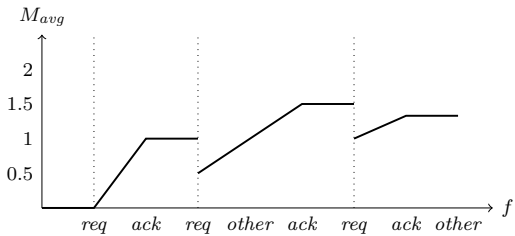
\begin{figure}[h]
	\centering
	\scalebox{0.922}{
		\begin{tikzpicture}
			\draw[->] (0,0) -- (6.5,0) node[anchor=west] {$f$};
			\draw
			(0.75,0) node[anchor=north] {$\req\vphantom{k}$}
			(1.5,0) node[anchor=north] {$\ack\vphantom{k}$}
			(2.25,0) node[anchor=north] {$\req\vphantom{k}$}
			(3,0) node[anchor=north] {$\other\vphantom{k}$}
			(3.75,0) node[anchor=north] {$\ack\vphantom{k}$}
			(4.5,0) node[anchor=north] {$\req\vphantom{k}$}
			(5.25,0) node[anchor=north] {$\ack\vphantom{k}$}
			(6,0) node[anchor=north] {$\other\vphantom{k}$};
			\draw[->] (0,0) -- (0,2.5) node[anchor=south] {$M_{\textit{avg}}$};
			\draw
			(0,0.5) node[anchor=east] {0.5}
			(0,1) node[anchor=east] {1}
			(0,1.5) node[anchor=east] {1.5}
			(0,2) node[anchor=east] {2};	
			\draw[thick] (0,0) -- (0.75,0) -- (1.5,1) -- (2.25,1);
			\draw[thick] (2.25,0.5) -- (3,1) -- (3.75,1.5) -- (4.5,1.5);
			\draw[thick] (4.5,1) -- (5.25,1.33) -- (6,1.33); 
			\draw[dotted] (0.75,0) -- (0.75,2.5);
			\draw[dotted] (2.25,0) -- (2.25,2.5);
			\draw[dotted] (4.5,0) -- (4.5,2.5);
		\end{tikzpicture}
	}
	\caption{Monitoring average response time for a trace $f$.}
	\label{fig:artVerdict}
\end{figure}

Now, let us consider an alphabet $\{\req_1,\ack_1,\req_2,\ack_2,\\\other\}$ with two types of matching $(\req_i,\ack_i)$ pairs.
The quantitative property $p_2$ measures the maximal response times for both pairs:
it maps every trace to an ordered pair of values from $\N \cup \{\infty\}$.
A construction similar to the one for $p_1$ gives us a precise monitor that uses 4 counters.
Indeed, we will show that 3 counters do not suffice to monitor $p_2$ precisely.
However, the quantitative property $p_2$ can be approximately monitored with 3 counters:
two counters can be used to store the maxima so far, and the third counter may track the current response
time prioritizing the pair $(\req_1,\ack_1)$ whenever both request types are active.
This 3-counter monitor will always under-approximate the maximal response time for the $(\req_2,\ack_2)$ pair. 
In case the resources are even scarcer, a 2-counter monitor can keep the same value as an under-approximation for both
maximal response times in one counter,
and use the second counter to wait sequentially for witnessing $(\req,\ack)$ pairs of both types before incrementing the
first counter. 
Just like the number of registers leads to precision-resource trade-offs for register monitors,
the number of states leads to precision-resource trade-offs for finite-state monitors.
For instance, a fixed number of states can encode counter values up to a certain magnitude,
but can under-approximate larger values.
We provide a general formal framework for quantitative and approximate monitoring which allows us to study such trade-offs
for different models of monitors.

In Section~\ref{sect:defn}, we define quantitative properties, approximate verdict functions, and how the precision of
monitors can be compared.
In Section~\ref{sect:quant}, we give a variety of different examples and closure operations for quantitative monitoring.
We also characterize the power of the important class of \emph{monotonic} monitors by showing that, in our framework,
the quantitative properties that can be monitored universally (on all traces) and precisely by monotonically increasing
verdict functions are exactly the co-continuous properties on the value domain.
In Section~\ref{sect:bool}, we embed several variations of the boolean value domain within our quantitative framework.
This allows us to characterize, within the safety-progress hierarchy \cite{Chang93}, which boolean properties can be
monitored universally and existentially; see Tables~\ref{tab:univ}~and~\ref{tab:exis}.
The section also connects our quantitative definitions of monitorability to the boolean definitions of
\cite{Falcone12,Aceto19a,Pnueli06b,Bauer11} and shows that our quantitative framework generalizes their popular boolean
settings conservatively.
Finally, in Section~\ref{sect:reg}, we present precision-resource trade-offs for register monitors.
For this purpose, we generalize the quantitative setting of \cite{Ferrere18,Ferrere20} to approximate monitoring within our framework.
In particular, we show a family of quantitative properties for which every additional counter register
improves the monitoring precision.

\smallskip
\noindent
\textbf{Related work.}
In the boolean setting, the first definition of \emph{monitorability} \cite{Kim02} focused on detecting violations of a
property.
This definition was generalized by \cite{Pnueli06b} and \cite{Bauer11} to capture satisfactions as well.
Later, instead of using a fixed, three-valued domain for monitoring, Falcone et al.~\cite{Falcone12} proposed a
definition with parameterized truth domains.
According to their definition, every linear-time property is monitorable in a four-valued domain where the usual
``inconclusive'' verdict is split into ``currently true'' and ``currently false'' verdicts.
Frameworks that capture existential as well as universal modalities for monitorability were studied in a
branching-time setting \cite{Francalanza17, Aceto19b}.

The prevalence of LTL and $\omega$-regular specifications in formal verification is also reflected in runtime
verification \cite{Bauer11, Bauer06, Bauer07}.
Recently, several more expressive models have been proposed, such as register monitors \cite{Ferrere18},
monitors for visibly pushdown languages \cite{Decker13}, quantified event automata \cite{Barringer12},
and many others for monitoring data events over an infinite alphabet of observations, as surveyed in \cite{Havelund18}.
One step towards quantitative properties is the augmenting of boolean specifications with quantities,
e.g., discounting and averaging modalities \cite{deAlfaro03,Bouyer14},
timed specifications \cite{Bonakdarpour13, Pinisetty17},
or specifications that include continuous signals, particularly in the context of cyber-physical systems
\cite{Bakhirkin17, Jakvsic18}.
Another prominent line of work that provides a framework for runtime verification beyond finite-state is that of Alur et al.~\cite{Alur17, Alur19, Alur20}.  
Their work focuses on runtime decidability issues for boolean specifications over streams of data events,
but they do not consider approximate monitoring at varying degrees of precision.
Quantitative frameworks for comparing traces and implementations for the same boolean
specification were studied in \cite{Caspi02, Bloem09}.
Our approach is fundamentally different as we consider quantitative property values.

Quantitative properties, a.k.a.\ quantitative languages, were defined in \cite{Chatterjee10}.
Although such properties have been studied much
in the context of probabilistic model checking \cite{Kwiatkowska07},
decision problems in verification \cite{Chatterjee10},
and games with quantitative objectives \cite{Bloem18,Bouyer18},
in runtime verification, we observe a gap.
While some formalisms for monitoring certain quantitative properties have been proposed
\cite{Ferrere20, Chatterjee16, Paul17},
to the best of our knowledge, our work is the first general semantic framework that explores what it means to monitor
and approximate generic quantitative properties of traces.
We believe that such a framework is needed for the systematic study of precision-resource trade-offs
in runtime verification. 
See \cite{Calinescu12} for a discussion of why quantitative verification at runtime is needed for self-adapting
systems, and \cite{Fulton18, Henzinger20b} for monitoring neural networks.

\section{Definitions} \label{sect:defn}

Let $\Sigma = \{a,b,\ldots\}$ be a finite alphabet of observations.
A \emph{trace} is a finite or infinite sequence of observations, denoted by $s,r,t \in \Sigma^*$ or $f,g,h \in \Sigma^\omega$, respectively.
For traces $w \in \Sigma^* \cup \Sigma^\omega$ and $s \in \Sigma^*$,
we write $s \prec w$ (resp.\ $s \preceq w$) iff
$s$ is a strict (resp.\ non-strict) finite prefix of $w$,
and denote by $\pref{w}$ the set of finite prefixes of $w$.

\subsection{Quantitative properties and verdict functions}

A \emph{boolean property} $P \subseteq \Sigma^\omega$ is a set of infinite traces, and
a \emph{value domain} $\D$ is a partially ordered set.
Unless otherwise stated, we assume that $\D$ is a complete lattice and, whenever appropriate,
we write 0, $-\infty$, and $\infty$ instead of $\bot$ and $\top$ for the least and greatest elements.
A \emph{quantitative property} $p: \Sigma^\omega \to \D$ is a function on infinite traces.
A \emph{verdict} $v: \Sigma^* \to \D$ is a function on finite traces such that
for all infinite traces $f \in \Sigma^\omega$, the set $\{v(s) : s\in \pref{f}\}$ of verdict values
over all prefixes of $f$ has a supremum (least upper bound) and an infimum (greatest lower bound).
If $\D$ is a complete lattice, then these limits always exist.
For an infinite trace $f \in \Sigma^\omega$, we write $v(f) = (v(s_i))_{i \in \N}$ for the infinite \emph{verdict sequence}
over the prefixes $s_i\prec f$ of increasing length $i$. 
We use the $\limsup$ or $\liminf$ of a verdict sequence $v(f)$ to represent the
``estimate'' that the verdict function $v$ provides for a quantitative property value $p(f)$
on the infinite trace $f$.

\begin{definition}
	Let $p$ be a quantitative property and $f \in \Sigma^\omega$ an infinite trace.
	A verdict function $v$ \emph{approximates $p$ on $f$ from below} (resp.\ \emph{above}) iff
	$\limsup v(f) \leq p(f)$ (resp.\ $p(f) \leq \liminf v(f)$).
	Moreover, $v$ \emph{monitors $p$ on $f$ from below} (resp.\ \emph{above}) iff
	the equality
	holds.
\end{definition}

\subsection{Universal, existential, and approximate monitorability}

We define three modalities of quantitative monitorability.

\begin{definition} \label{defn:univmon}
	A quantitative property $p$ is \emph{universally monitorable from below} (resp.\ \emph{above}) iff
	there exists a verdict function $v$ such that for every $f \in \Sigma^\omega$ we have that
	$v$ monitors $p$ on $f$ from below (resp.\ above).
\end{definition}

\begin{definition} \label{defn:exismon}
	A quantitative property $p$ is \emph{existentially monitorable from below} (resp.\ \emph{above}) iff
	there exists a verdict function $v$ such that
	(i) for every $f \in \Sigma^\omega$ we have that $v$ approximates $p$ on $f$ from below
	(resp.\ above), and
	(ii) for every $s \in \Sigma^*$ there exists $f \in \Sigma^\omega$ such that
	$v$ monitors $p$ on $sf$ from below (resp.\ above).
\end{definition}

\begin{definition}
	A quantitative property $p$ is \emph{approximately monitorable from below} (resp.\ \emph{above}) iff
	there exists a verdict function $v$ such that for every $f \in \Sigma^\omega$ we have that
	$v$ approximates $p$ on $f$ from below (resp.\ above).
\end{definition}

Observe that every property is trivially approximately monitorable from below or above.
We demonstrate the definitions in the example below.

\begin{example} \label{ex:defns}
	Let $\Sigma = \{\req_1,\ack_1,\req_2,\ack_2,\other\}$ and 
	$\D$ be the nonnegative integers with $\infty$.
	Consider the maximal response-time properties $p_1$ and $p_2$ over $(\req_1,\ack_1)$ and $(\req_2,\ack_2)$ pairs,
	respectively.
	For every $f \in \Sigma^\omega$, let $p(f) = \max(p_1(f),p_2(f))$.
	Consider the verdict $v_1$ that counts both response times and outputs the maximum of the two,
	the verdict $v_2$ that counts and computes the maximum only for the $(\req_1,\ack_1)$ pair, and
	the constant verdict $v_3$ that always outputs 0.
	Evidently, $v_1$ universally monitors $p$ from below, and $v_3$ approximately monitors $p$ from below.
	Moreover, $v_2$ existentially monitors $p$ from below because the true maximum can only be greater, and
	we can extend every finite trace
	$s \in \Sigma^*$ with $f = \req_1 \cdot \other^\omega$ such that $\limsup v_2(sf) = p(sf) = \infty$.
\end{example}

\subsection{Monotonic verdict functions}

Of particular interest are \emph{monotonic} verdict functions, because the ``estimates'' they provide
for a quantitative property value are always conservative (below or above) and can improve in quality over time.
On the other hand, some properties, such as average response time, inherently require non-monotonic
verdict functions for universal monitoring.

\begin{definition}
	A verdict function $v$ is \emph{monotonically increasing} (resp.\ \emph{decreasing}) iff
	for every $s,t \in \Sigma^*$ we have 
	$s \prec t$ implies $v(s) \leq v(t)$ (resp.\ $v(s) \geq v(t)$).
	Moreover, $v$ is \emph{monotonic} iff
	it is either monotonically increasing or monotonically decreasing.
	If $v$ is monotonic or non-monotonic, then it is \emph{unrestricted}.
\end{definition}

If the value domain $\D$ has a least and a greatest element,
every monotonic verdict $v$ that universally monitors a property $p$ from below also universally monitors $p$ from above.
Therefore, in such cases, we say that $v$ \emph{universally monitors} $p$.
In Example~\ref{ex:defns} above, the verdict $v_1$ is monotonically increasing and thus universally monitors $p$.
Let $v_4$ be such that $v_4(s) = \infty$ if $s$ contains a request that is not
acknowledged, and $v_4(s) = v_1(s)$ otherwise.
The verdict $v_4$ is not monotonic, but it universally monitors $p$ from above.

\subsection{Comparison of verdict functions}

Quantitative monitoring provides a natural notion of precision for verdict functions.

\begin{definition} \label{defn:precise}
	Let $p$ be a quantitative property that is (universally, existentially, or approximately) monitorable
	from below (resp.\ above) by the verdict functions $v_1$ and $v_2$.
	The verdict $v_1$ is \emph{more precise} than the verdict $v_2$ iff for every $f \in \Sigma^\omega$
	we have $\limsup v_2(f) \leq \limsup v_1(f)$ (resp.\ $\liminf v_1(f) \leq \liminf v_2(f)$) and
	there exists $g \in \Sigma^\omega$ such that $\limsup v_2(g) < \limsup v_1(g)$
	(resp.\ $\liminf v_1(g) < \liminf v_2(g)$).
	Moreover, $v_1$ and $v_2$ are \emph{equally precise} iff for every $f \in \Sigma^\omega$
	we have $\limsup v_2(f) = \limsup v_1(f)$ (resp.\ $\liminf v_1(f) = \liminf v_2(f)$).
\end{definition}

Note that for a quantitative property $p$, if the verdict functions $v_1$ and $v_2$ universally monitor $p$ both
from below or from above, then $v_1$ and $v_2$ are equally precise.
Two monotonically increasing or monotonically decreasing verdict functions can be compared not only according to
their precision but also according to their \emph{speed}, that is, how quickly they approach the property value.
This will be important if monitors have limited resources and their outputs are delayed, i.e.,
they affect not the current but a future verdict value.

\section{Monitorable Quantitative Properties} \label{sect:quant}
\subsection{Examples}

We provide several examples of quantitative properties and investigate their monitorability.

\begin{example} [Maximal response time] \label{ex:mrt}
	Let $\Sigma = \{\req,\ack,\\\other\}$ and $\D = \N \cup \{\infty\}$.
	Let $\mrt : \Sigma^* \to \D$ be such that $\mrt(s) = \infty$ if, in $s$, a $\req$ is followed by another $\req$ without an $\ack$ in between, it equals the maximal number $m_s$ of observations between matching $(\req,\ack)$ pairs if there is no pending request in $s$, and otherwise it equals $\max(m_s,n)$ where $n$ is the current response time.
	For every $f \in \Sigma^\omega$, let us denote by $\mrt(f)$ the infinite sequence $(\mrt(s_i))_{i \in \N}$ over the prefixes $s_i \prec f$ of increasing length $i$.	
	Consider the property $p(f) = \lim \mrt(f)$ that specifies the maximal response time of a server that can process
	at most one request at a time.
	To monitor $p$, we use $\mrt$ as the verdict, i.e., we let $v(s) = \mrt(s)$ for every $s \in \Sigma^*$.
	Observe that $\mrt$ is monotonically increasing, and the construction yields $\lim v(f) = p(f)$ for every $f \in \Sigma^\omega$.
	Therefore, the verdict $v$ universally monitors $p$.
\end{example}

The maximal response-time property of Example~\ref{ex:mrt} is evidently
infinite-state because it requires counting up to an arbitrarily large integer.
However, there are finite-state approximations that improve in precision with every additional state. 
We say that a finite-state machine \emph{generates} a verdict function
iff, on every finite trace, the machine's output equals the verdict value,
where an output is a mapping from the set of states to the value domain.

\begin{example}	[Approximate monitoring of maximal response time]
	Consider the maximal response-time property $p$ from Example~\ref{ex:mrt}.
	Let $M_k$ be a finite-state machine with $k$ states,
	and let $v_k$ be the verdict generated by $M_k$.
	For every $k \in \N$, the best the verdict $v_k$ can do is to approximately monitor $p$ from below,
	because it can only count up to some integer $m \leq k$.
	Suppose that we are given $k+1$ states.
	We can use the additional state to construct a machine $M_{k+1}$ from $M_k$ to generate
	a more precise verdict $v_{k+1}$ as follows.
	We add the appropriate transitions from the states that have the output value of $m$ to the new state, which is assigned the output $m+1$.
	With the additional transitions, the machine $M_{k+1}$ can continue counting for one more
	step after reading a trace in which the current maximum is $m$.
	Therefore, $v_{k+1}$ is more precise than $v_k$.	
\end{example}

Next, we define the average response-time property and present two verdict functions
that illustrate another kind of precision-resource trade-off for monitors.

\begin{example} [Average response time] \label{ex:art}
	Let $\Sigma = \{\req,\ack,\\\other\}$ and $\D = \R \cup \{\infty\}$.
	Let $\art : \Sigma^* \to \D$ be such that $\art(s) = \infty$ if $s$ contains
	a $\req$ followed by another $\req$ without an $\ack$ in between,
	it equals the average number of observations between matching $(\req,\ack)$ pairs if there is no pending $\req$ in $s$,
	and otherwise it equals $\frac{n \cdot x_n + m}{n+1}$,
	where $n$ is the number of acknowledged requests, $x_n$ is the average response time for the
	first $n$ requests, and $m$ is the number of observations since the last $\req$.
	For every $f \in \Sigma^\omega$,	
	let  $\art(f) = (\art(s_i))_{i \in \N}$ over the prefixes $s_i \prec f$ of increasing length $i$.	
	Now, define $\limavg(f) = \liminf \art(f)$ for every $f \in \Sigma^\omega$ \cite{Chatterjee10},
	and let $p$ be the quantitative property such that $p(f) = \limavg(f)$.
	In other words, $p$ specifies the average response time of a server that can process at most one
	request at a time.
	To monitor $p$, we can use the function $\art$ as a verdict, i.e.,
	let $v$ be such that $v(s) = \art(s)$ for all $s \in \Sigma^*$.
	Intuitively, the moving average approaches to the property value as $v$ observes longer prefixes.
	Therefore, by construction, for every $f \in \Sigma^\omega$, we have $\liminf v(f) = p(f)$,
	which means that $p$ is universally monitorable from above by an unrestricted verdict function.
	
	Alternatively, we can use the monotonic verdict function $v'$ that universally monitors
	the maximal response-time property in Example~\ref{ex:mrt}.
	Observe that $v'$ existentially monitors $p$ from above because
	(i) the maximal response time of a trace is greater than its average response time, and
	(ii) for every finite prefix $s$ there is an extension $f$ that contains a request that is not acknowledged,
	which yields $\lim v'(sf) = p(sf) = \infty$.
\end{example}

Boolean safety and co-safety properties can be embedded in a quantitative setting
by considering their \emph{discounted} versions \cite{deAlfaro03}.
We show that discounted safety and co-safety properties are universally monitorable.

\begin{example} [Discounted safety and co-safety]
	Let $p$ be a discounted safety property, that is,
	$p(f) = 1$ if $f$ does not violate the given safety property, and $p(f) = 1 - \frac{1}{2^n}$ if the shortest violating prefix of $f$ has length $n$.
	Similarly, let $q$ be a discounted co-safety property:
	$q(f) = 0$ if $f$ does not satisfy the given co-safety property, and  $q(f) = \frac{1}{2^n}$ if the shortest satisfying prefix of $f$ has length $n$.
	To monitor these two properties, we use verdict functions $v_p$ and $v_q$ that work
	similarly as $p$ and $q$ on finite traces, that is,
	$v_p(s) = 1$ if $s$ is not violating for the given safety property, and $v_p(s) = 1 - \frac{1}{2^n}$ if the shortest violating prefix of $s$ has length $n$; and similarly for $v_q$.
	One can easily verify that $p$ and $q$ are universally monitorable by $v_p$ and $v_q$, respectively.
\end{example}

Finally, we look at another classical value function for quantitative properties,
often called \emph{energy} values \cite{Chatterjee10}.

\begin{example}	[Energy]
	Let $A = (Q, \Sigma, \delta, q_0, w)$ be a deterministic finite automaton with weighted transitions,
	where $Q$ is a set of states, $\Sigma$ is an alphabet, $\delta \subseteq Q \times \Sigma \times Q$ is a set of transitions,
	$q_0$ is the initial state, and $w : \delta \to \Z$ is a weight function.
	Let $s = \sigma_1 \ldots \sigma_n$ be a finite trace of length $n$,
	and let $q_0 \ldots q_n$ be the corresponding run of $A$.
	We define $A(s) = \sum_{i=1}^{n} w(q_{i-1}, \sigma_i, q_i)$, where $A(\varepsilon) = 0$.
	Consider the value domain $\D = \Z \cup \{\infty\}$.
	Let $p$ be a property such that, for every $f \in \Sigma^\omega$, we have $p(f) = k$ where $k$ is the smallest nonnegative value that satisfies $A(s) + k \geq 0$ for every finite prefix $s \prec f$.
	To monitor $p$, we construct the following verdict function:
	given  $s \in \Sigma^*$, let $v(s) = -\min\{A(r) \mid r \in \pref{s}\}$.
	Note that $v$ is monotonically increasing.
	On an infinite trace $f \in \Sigma^\omega$, if $v(f)$ approaches $\infty$,
	then $f$ yields a negative-weight loop on $A$, therefore
	$p(f) = \infty$.
	Otherwise, if $v(f)$ converges to a finite value, then it is equal to $p(f)$ by construction,
	which means that $v$ universally monitors $p$.
\end{example}

\subsection{Closure under operations on the value domain}

Let $\D$ be a value domain and $p : \Sigma^\omega \to \D$ be a quantitative property.
We define the \emph{inverse of $\D$}, denoted $\D_{\textit{inv}}$,
as the value domain that contains the same elements as $\D$ with reversed ordering.
Moreover, we define the \emph{complement of $p$} as $\overbar{p} : \Sigma^\omega \to \D_{\textit{inv}}$
such that $\overbar{p}(f) = p(f)$.

\begin{proposition}
	A quantitative property $p$ is universally (resp.\ existentially; approximately)
	monitorable from below iff $\overbar{p}$ is universally (resp.\ existentially; approximately)
	monitorable from above.
\end{proposition}

If the value domain $\D$ is a lattice, then monitorability from below is preserved by the least upper bound
(written $\max$) and from above by greatest lower bound (written $\min$).
For all quantitative properties $p$ and $q$ on $\D$, and all infinite traces $f\in\Sigma^\omega$, let $\max(p,q)(f)=\max(p(f),q(f))$ and $\min(p,q)(f)=\min(p(f),q(f))$.

\begin{proposition}
	For all quantitative properties $p$ and $q$ on a lattice,
	if $p$ and $q$ are universally (resp.\  existentially; approximately) monitorable from below
	(resp.\ above),
	then the property $\max(p,q)$ (resp.\ $\min(p,q)$) is also universally (resp.\ existentially; approximately)
	monitorable from below (resp.\ above).
\end{proposition}

\begin{proof}
	Let $v_p$ and $v_q$ be two verdict functions that universally monitor $p$ and $q$ from below.
	Then, we have $\max(p(f),q(f)) = \max(\limsup v_p(f), \limsup v_q(f))$ for every $f \in \Sigma^\omega$.
	Since we assume that the domain contains a greatest element, for every $f \in \Sigma^\omega$, we also have $\max(\limsup v_p(f), \limsup v_q(f))$ equals $\limsup(\max(v_p(f),v_q(f)))$.
	Therefore, we can use $\max(v_p,v_q)$ as a verdict function to universally monitor $\max(p,q)$ from below the same way $v_p$ and $v_q$ monitor $p$ and $q$.
	The case for $\min$ is symmetric, and the cases for existential and approximate monitoring can be proved similarly by using the fact that the domain is a lattice.
\end{proof}

\begin{proposition}
	For all quantitative properties $p$ and $q$ on a lattice, if $p$ and $q$ are (universally, existentially, or approximately) monitorable from below (resp. above), the property $\min(p,q)$ (resp. $\max(p,q)$) is approximately monitorable from below (resp. above).
\end{proposition}

\begin{proof}
	Let $v_p$ and $v_q$ be verdict functions that monitor $p$ and $q$ from below, therefore for every infinite trace $f \in \Sigma^\omega$ we have $\limsup v_p(f) \leq p(f)$ and $\limsup v_q(f) \leq q(f)$.
	Because $\limsup \min(v_p(f),v_q(f)) \leq \min(\limsup v_p(f), \limsup v_q(f))$ for every $f \in \Sigma^\omega$, we can use $\min(v_p,v_q)$ as a verdict function to approximately monitor $\min(p,q)$ from below.
	The case for $\max$ is dual.
\end{proof}

If $\D$ is a numerical value domain with addition and multiplication, such as the reals or integers, or their
nonnegative subsets, then not all modalities of monitorability are preserved under these operations.
For all quantitative properties $p$ and $q$ on $\D$, and all infinite traces $f\in\Sigma^\omega$,
let $(p+q)(f)=p(f)+q(f)$ and $(p \cdot q)(f) =p(f) \cdot q(f)$.
Since $\limsup$ is subadditive and submultiplicative while $\liminf$ superadditive and supermultiplicative, one can easily conclude the following.

\begin{proposition}
	For all quantitative properties $p$ and $q$ on a numerical value domain,
	if $p$ and $q$ are (universally, existentially, or approximately) monitorable from below
	(resp.\ above),
	then $p+q$ and $p \cdot q$ are approximately monitorable from below
	(resp.\ above).
\end{proposition}

However, monitorability is preserved under any monotonically increasing continuous function on value domains that are totally ordered.

\begin{proposition}
	Let $\D$ be a totally-ordered value domain.
	Consider a quantitative property $p : \Sigma^\omega \to \D$ and a monotonically increasing continuous function $\phi : \D \to \D$.
	If $p$ is (universally, existentially, or approximately) monitorable from below (resp. above),
	then so is $\phi(p)$.
\end{proposition}

\subsection{Continuous quantitative properties}

For this section, we assume that $\D$ is a complete lattice and
define \emph{continuous} and \emph{co-continuous} properties on $\D$.
Let $p$ be a quantitative property and, for every $s \in \Sigma^*$, let $\nu_p(s) = \sup\{p(sf) \mid f \in \Sigma^\omega\}$.
For $f \in \Sigma^\omega$,
the function $\nu_p$ generates an infinite sequence
$\nu_p(f) = (\nu_p(s_i))_{i \in \N}$ over the prefixes $s_i\prec f$ of increasing length $i$. 
Similarly, let $\mu_p(s) = \inf\{p(sf) \mid f \in \Sigma^\omega\}$ and
extend it to generate infinite sequences on infinite traces.

\begin{definition} [\cite{Weihrauch87}] \label{defn:cont}
	A property $p$ is \emph{continuous} iff
	for every infinite trace $f \in \Sigma^\omega$, we have $p(f) = \lim \nu_p(f)$.
	Moreover, $p$ is \emph{co-continuous} iff $\overbar{p}$ continuous, or equivalently, iff $p(f) = \lim \mu_p(f)$
	for every $f \in \Sigma^\omega$.
\end{definition}

Intuitively, the continuous and co-continuous properties constitute well-behaved sets of properties in the sense that,
to monitor them, there is no need for speculation.
For example, considering a continuous property, the least upper bound can only decrease after reading longer prefixes;
therefore, a verdict function monitoring such a property can simultaneously be conservative and precise.
We make this connection more explicit and show that continuous and co-continuous properties satisfy the desirable property
of being universally monitorable by monotonic verdict functions.

\begin{theorem} \label{thm:continuous}
	A quantitative property $p$ is continuous iff
	$\overbar{p}$ is universally monitorable by a monotonically increasing verdict function.
\end{theorem}

\begin{proof}
	Observe that $\overbar{p}$ is universally monitorable by a monotonically increasing verdict function iff $p$ is universally monitorable by a monotonically decreasing verdict function.
	For the \emph{only if} direction, suppose $p$ is continuous, i.e., $\lim \nu_p(f) = p(f)$ for every $f \in \Sigma^\omega$.
	Since $\nu_p$ is monotonically decreasing and it converges to the property value for every infinite trace, we can use it as the verdict function to universally monitor $p$.
	
	Now, let $v$ be a monotonically decreasing verdict function such that $\lim v(f) = p(f)$ for all $f \in \Sigma^\omega$.
	We claim that $v(s) \geq \nu_p(s)$ for all $s \in \Sigma^*$.
	Suppose towards contradiction that $v(s) < \nu_p(s) $ for some $s \in \Sigma^*$.
	Since we have either (i) $\nu_p(s) = p(sg)$ for some $g \in \Sigma^\omega$, or (ii) for every $g \in \Sigma^\omega$ there exists $h \in \Sigma^\omega$ such that $p(sg) < p(sh)$, we obtain $v(s) < p(sf)$ for some $f \in \Sigma^\omega$. 
	It contradicts the assumption that $v$ is a monotonically decreasing verdict which universally monitors $p$ from below, therefore our claim is correct.
	Now, observe that $v(s) \geq \nu_p(s)$ for all $s \in \Sigma^*$ implies $\lim v(f) \geq \lim \nu_p(f)$ for all $f \in \Sigma^\omega$.
	Since $v$ universally monitors $p$, we get $p(f) \geq \lim \nu_p(f)$ for all $f \in \Sigma^\omega$.
	By the definition of $\nu_p$, we also know that for every property $p$ and infinite trace $f \in \Sigma^\omega$, we have $\lim \nu_p(f) \geq p(f)$.
	Therefore, we conclude that $\lim \nu_p(f) = p(f)$ for all $f \in \Sigma^\omega$, i.e., $p$ is continuous.
\end{proof}

Combining Theorem~\ref{thm:continuous} and Definition~\ref{defn:cont}, we immediately get the following characterization for the co-continuous properties.

\begin{corollary} \label{cor:cocontinuous}
	A quantitative property $p$ is co-continuous iff
	$p$ is universally monitorable by a monotonically increasing verdict function.
\end{corollary}

Let $\D$ be a numerical domain and recall the maximal response-time property from Example~\ref{ex:mrt}.
As we discussed previously, it is universally monitorable by a monotonically increasing verdict function,
and therefore co-continuous.
By the same token, one can define the \emph{minimal response-time} property, which is continuous.
However, average response time, which requires a non-monotonic verdict function although it is universally monitorable from above,
is neither continuous nor co-continuous.
We also remark that discounted safety and co-safety properties \cite{deAlfaro03} are continuous and co-continuous, respectively.
In Section~\ref{sect:bool}, we will discuss how these notions relate to safety and co-safety in the boolean setting.

\section{Monitoring Boolean Properties} \label{sect:bool}

\subsection{Boolean monitorability as quantitative monitorability}

Quantitative properties generalize boolean properties.
For every boolean property $P \subseteq \Sigma^\omega$, the characteristic function
$\tau_P : \Sigma^\omega \to \{\fa, \tr\}$ is a quantitative property,
where $\tau_P(f) = \tr$ if $f \in P$, and $\tau_P(f) = \fa$ if $f \notin P$.
Using this correspondence, we can embed the main boolean notions of monitorability
within our quantitative framework.
For this, we consider four different boolean value domains:
\begin{itemize}
	\item $\B = \{\fa, \tr\}$ such that $\fa$ and $\tr$ are incomparable.
	\item $\B_\bot = \B \cup \{\bot\}$ such that $\bot < \fa$ and $\bot < \tr$.
	\item $\B_t = \{\fa, \tr\}$ such that $\fa < \tr$.
	\item $\B_f = \{\fa, \tr\}$ such that $\tr < \fa$.
\end{itemize}

Most work in monitorability assumes irrevocable verdicts.
On the domains $\B$ and $\B_\bot$, where $\tr$ and $\fa$ are incomparable,
the irrevocability of verdicts corresponds to monotonically increasing verdict functions.
For these, positive verdicts in $\B_t$ and negative verdicts in $\B_f$ are also irrevocable.
The following observations about verdict functions on boolean domains are useful as well.

\begin{remark} \label{rem:allornothing}
	Let $v$ be a verdict function on $\B$ or $\B_\bot$.
	If $v$ is monotonic, it cannot switch between $\tr$ and $\fa$, as these values are incomparable.
	Therefore, $v$ can monitor only $\emptyset$ and $\Sigma^\omega$ in $\B$.
	If $v$ is unrestricted, it can switch between $\tr$ and $\fa$ only finitely often,
	because the $\limsup$ and $\liminf$ over every infinite trace must be defined.
\end{remark}

We begin with the classical definition of monitorability for boolean properties \cite{Pnueli06b,Bauer11}.
Let $P \subseteq \Sigma^\omega$ be a boolean property.
A finite trace $s \in \Sigma^*$ \emph{positively} (resp.\ \emph{negatively}) \emph{determines} $P$ iff
for every $f \in \Sigma^\omega$, we have $sf \in P$ (resp.\ $sf \notin P$).
The boolean property $P$ is \emph{classically monitorable} iff for every $s \in \Sigma^*$,
there exists $r \in \Sigma^*$ such that $sr$ positively or negatively determines $P$.
This definition
coincides with the \emph{persistently informative monitorability} of \cite{Aceto19a}.
It is also captured by our definition of existential monitorability by monotonic verdicts on $\B_\bot$.

\begin{proposition} \label{thm:strmon}
	A boolean property $P$ is classically monitorable iff $\tau_P$ is existentially monitorable
	from below by a monotonically increasing verdict function on $\B_\bot$.
\end{proposition}

According to \cite{Aceto19a}, a boolean property $P$ is \emph{satisfaction} (resp.\ \emph{violation})
\emph{monitorable} iff there exists a monitor that reaches a positive (resp.\ negative) verdict for
every $f \in P$ (resp.\ $f \notin P$).
More generally, if monitorability is parameterized by a truth domain as in \cite{Falcone12},
then violation and satisfaction monitorability correspond to monitorability over $\{\bot,\fa\}$ and
$\{\bot,\tr\}$, and capture exactly the classes of safety and co-safety properties, respectively.
In our framework, violation (resp.\ satisfaction) monitorability is equivalent to
universal monitorability by monotonically increasing verdicts on $\B_f$ (resp.\ $\B_t$),
because they require reaching an irrevocable negative (resp.\ positive) verdict for traces
that violate (reps.\ satisfy) the property.

\begin{theorem} [\cite{Falcone12}] \label{thm:safe+cosafe} 
	A boolean property $P$ is safe (resp.\ co-safe) iff $\tau_P$ is universally monitorable by a monotonically increasing verdict function on $\B_f$ (resp.\ $\B_t$).
\end{theorem}

A boolean property $P$ is \emph{partially monitorable} according to \cite{Aceto19a} iff
it is satisfaction or violation monitorable.
This corresponds to parametric monitorability over the 3-valued domain $\{\bot,\tr,\fa\}$,
and is equivalent to the union of safety and co-safety \cite{Falcone12}. 
Due to the duality of $\B_f$ and $\B_t$, in our framework, partial monitorability corresponds to universal monitorability by monotonic verdict functions on either of these domains.

\begin{corollary} \label{thm:safeorcosafe}
	A boolean property $P$ is safe or co-safe iff $\tau_P$ is universally
	monitorable by a monotonic verdict function on $\B_t$ (equivalently, on $\B_f$).
\end{corollary}

Also defined in \cite{Aceto19a} is the notion of \emph{complete monitorability},
which requires both satisfaction and violation monitorability.
It is equivalent to our universal monitorability by monotonic verdict functions
on $\B_\bot$, meaning that for every trace $f \in P$ we reach a positive verdict,
and for every $f \notin P$, a negative verdict.

\begin{theorem} [\cite{Aceto19a}] \label{thm:safencosafe}
	A boolean property $P$ is both safe and co-safe iff $\tau_P$ is universally
	monitorable by a monotonically increasing verdict function on $\B_\bot$.
\end{theorem}

Based on the idea of revocable verdicts, a 4-valued domain $\{\tr,\fa,\tr_c,\fa_c\}$ is also considered in \cite{Falcone12}, where $\tr$ and $\fa$ are still irrevocable
but the inconclusive verdict $\bot$ is split into two verdicts $\tr_c$ (``currently true'')
and $\fa_c$ (``currently false'') for more nuanced reasoning on finite traces.
In the universe of $\omega$-regular properties, their monitorability over this domain corresponds to the class of reactivity properties of the safety-progress hierarchy \cite{Chang93}.
In our framework, unrestricted verdict functions provide a similar effect as revocable verdicts.
We will show in Theorem~\ref{thm:react} and Example~\ref{ex:react} that unrestricted verdict functions on $\B_\bot$ can existentially monitor reactivity properties and more.

Finally, two weak forms of boolean monitorability defined in \cite{Aceto19a} are
\emph{sound monitorability} and \emph{informative monitorability}.
While sound monitorability corresponds to approximate monitorability in $\B_\bot$,
informative monitorability corresponds to approximate monitorability in $\B_\bot$ by monotonic verdicts
but excluding the constant verdict function $\bot$.

\subsection{Monitoring the safety-progress hierarchy}

We first show that some of the modalities of quantitative monitoring are equivalent over boolean domains.
Proposition~\ref{pro:equiv} also indicates some limitations of flat value domains.

\begin{proposition} \label{pro:equiv}
	Let $P$ be a boolean property and $\tau_P$ be the corresponding quantitative property.
	The following statements are equivalent.
	\begin{enumerate}[(1)]
		\item $\tau_P$ is existentially monitorable from below by an unrestricted verdict function on $\B$.
		\item $\tau_P$ is universally monitorable from below by an unrestricted verdict function on $\B$.
		\item $\tau_P$ is universally monitorable from below by an unrestricted verdict function on $\B_\bot$.
	\end{enumerate}
\end{proposition}

\begin{proof}
	The key observation for the proofs is that $\tr$ and $\fa$ are incomparable in $\B$ and $\B_\bot$, as pointed out in Remark~\ref{rem:allornothing}.
	
	$(1) \iff (2)$:
	Let $\tau_P$ be existentially monitorable from below by a verdict function $v$ on $\B$.
	Since $\tr$ and $\fa$ are incomparable, for every $f \in \Sigma^\omega$, if $\limsup v(f) \leq \tau_P(f)$ then $\limsup v(f) = \tau_P(f)$ in domain $\B$.
	Therefore, $v$ also universally monitors $p$ from below in $\B$.
	The other direction follows from Definitions~\ref{defn:univmon}~and~\ref{defn:exismon}.
	
	$(2) \iff (3)$:	
	The \emph{only if} direction follows from the fact that $\B_\bot$ is an extension of $\B$ with a least element.
	For the \emph{if} direction, suppose $v$ is a verdict function that universally monitors $\tau_P$ from below in $\B_\bot$.
	We construct a verdict function $u$ that imitates $v$ in $\B$ as follows: let $u(s) = v(s)$ if $v(s) \neq \bot$, and $u(s) = v(r)$ otherwise, where $r$ is the longest prefix of $s$ such that $v(r) \neq \bot$ (if there is no such prefix, assume w.l.o.g. that $u(s) = \tr$).
	Now, let $f \in \Sigma^\omega$ be an infinite trace, and observe that whenever $v(f)$ converges, so does $u(f)$.
	If $v(f)$ does not converge, then the subsequential limits must be either (i) $\bot$ and $\tr$, or (ii) $\bot$ and $\fa$.
	Suppose (i) is true.
	Then, there exists a prefix $s \prec f$ such that for all $r \in \Sigma^*$ satisfying $sr \prec f$ we have $v(sr) = \bot$ or $v(sr) = \tr$.
	If $v(s) = \tr$, then, by construction, $u$ always outputs $\tr$ starting from $s$.
	Otherwise, $u$ outputs $\fa$ until $v$ outputs $\tr$ (which is bound to happen by supposition), and converges to $\tr$ afterwards.
	The case for (ii) is dual.
	Therefore, for every $f \in \Sigma^\omega$, we have $\limsup u(f) = \limsup v(f)$, which means that $u$ universally monitors $p$ from below.
\end{proof}

\begin{proposition} \label{pro:exmonbtbf}
	For every boolean property $P$, we have that $\tau_P$ is existentially monitorable from below by a
	monotonically increasing verdict function on $\B_t$ iff $\tau_P$ is existentially monitorable from
	below by a monotonically increasing verdict function on $\B_f$.
\end{proposition}

\begin{proof}
	Suppose $\tau_P$ is existentially monitorable from below by a monotonically increasing verdict function $v$ on $\B_t$.
	Consider the following verdict function: $u(s) = \fa$ if $v(s) = \fa$ and $\tau_P(sf) = \fa$ for all $f \in \Sigma^\omega$; and $u(s) = \tr$ if $v(s) = \tr$ or $\tau_P(sf) = \tr$ for some $f \in \Sigma^\omega$.
	Notice that for every $s \in \Sigma^*$, if $v(s) = \tr$ then $\tau_P(sf) = \tr$ for all $f \in \Sigma^\omega$, and if $v(s) = \fa$ then $\tau_P(sf) = \fa$ for some $f \in \Sigma^\omega$, or $v(sr) = \tr$ for some $r \in \Sigma^*$.
	Therefore, we can equivalently formulate $u$ as follows: $u(s) = \fa$ if $\tau_P(sf) = \fa$ for all $f \in \Sigma^\omega$; and $u(s) = \tr$ if $\tau_P(sf) = \tr$ for some $f \in \Sigma^\omega$.
	The function $u$ is indeed monotonically increasing.
	Further, $\limsup u(f) = \fa$ implies that there is $s \prec f$ such that $\tau_P(sg) = \fa$ for every $g \in \Sigma^\omega$, which means that for every $f \in \Sigma^\omega$ we have $\limsup u(f) \leq \tau_P(f)$.
	
	Next, we show that for every $s \in \Sigma^*$ there exists $f \in \Sigma^\omega$ such that $\limsup u(sf) = \tau_P(sf)$.
	Suppose towards contradiction that for some $s \in \Sigma^*$ every $f \in \Sigma^\omega$ gives us $\limsup u(sf) < \tau_P(sf)$, i.e., $\limsup u(sf) = \tr$ and $\tau_P(sf) = \fa$.
	Since $\limsup u(sf) = \tr$ and $u$ is monotonically increasing, we get $u(r) = \tr$ for every $r \prec sf$.
	It means that, by construction, for every $r \prec sf$ there exists $g \in \Sigma^\omega$ such that $\tau_P(rg) = \tr$.
	However, we get a contradiction since $s \prec sf$ and $\tau_P(sf) = \fa$ for every $f \in \Sigma^\omega$ by supposition.
	Therefore, we conclude that $u$ existentially monitors $\tau_P$ from below in $\B_f$.
	The \emph{if} direction can be proved symmetrically.
\end{proof}

Before we relate various modalities of monitoring and boolean value domains to the rest of the
safety-progress classification of boolean properties \cite{Chang93},
we discuss how boolean safety and co-safety are special cases of continuous and co-continuous properties from Section~\ref{sect:quant}.
Consider the value domain $\B_t$,
let $P$ be a safety property and $\tau_P$ be the corresponding quantitative property.
Observe that for every $s \in \Sigma^*$, we have $\nu_{\tau_P}(s) = \fa$ if $s$ negatively determines $P$,
and $\nu_{\tau_P}(s) = \tr$ otherwise.
Since $P$ is safe, we also have $\tau_P(f) = \lim \nu_{\tau_P}(f)$ for every $f \in \Sigma^\omega$,
which means that $\tau_P$ is continuous.
Moreover, the inverse $\overbar{\tau_P}$ is a co-continuous property on $\B_f$,
and it still corresponds to the same boolean safety property.
Therefore, by Theorem~\ref{thm:continuous}, we get that property $P$ is safe iff
$\overbar{\tau_P}$ is universally monitorable by a monotonically increasing verdict function on $\B_f$.
Similarly, co-safety properties correspond to the co-continuous properties on $\B_t$;
and thus, they are exactly the properties that are universally monitorable by monotonically increasing verdict functions
on $\B_t$.

Positive, finite boolean combinations of safety and co-safety properties are called \emph{obligation} properties \cite{Chang93}.
Every obligation property $P$ can be expressed in a canonical conjunctive normal form $\bigcap_{i=1}^n (S_i \cup C_i)$ for some positive integer $n$, where $S_i$ is safe and $C_i$ is co-safe for all $1 \leq i \leq n$.
Moreover, an obligation property in conjunctive normal form with $n = k$ is a \emph{$k$-obligation property}.

We prove that obligation properties are universally monitorable in $\B$, which naturally requires finitely many switches
between verdicts $\tr$ and $\fa$.
Moreover, we establish an equivalence between the infinite hierarchy of obligation properties and a hierarchy of verdict
functions on $\B$.

\begin{theorem} \label{thm:obl}
	A boolean property $P$ is a $k$-obligation property iff $\tau_P$ is universally monitorable by a verdict function on $\B$ that changes its value at most $2k$ times.
\end{theorem}

\begin{proof}
	Suppose $P$ is a $k$-obligation property, in other words, $P = \bigcap_{i=1}^k (S_i \cup C_i)$ for some integer $k \geq 1$ where $S_i$ is safe and $C_i$ is co-safe for each $1 \leq i \leq k$.
	Consider the following verdict function: $v(s) = \tr$ if for every $1 \leq i \leq k$ we have $s$ does not negatively determine $S_i$ or $s$ positively determines $C_i$; and $v(s) = \fa$ if there exists $1 \leq i \leq k$ such that $s$ negatively determines $S_i$ and $s$ does not positively determine $C_i$.
	Note that if a finite trace $s$ positively or negatively determines a boolean property, then so does $sr$ for every finite continuation $r$.
	If $P$ cannot be expressed as a $(k-1)$-obligation property, then
	there exists a sequence of finite traces $s_1 \prec r_1 \prec \ldots \prec s_k \prec r_k$, w.l.o.g., such that for every $1 \leq i \leq k$ trace $s_i$ negatively determines $S_i$, does not negatively determine any $S_j$ for $j > i$, and does not positively determine any $C_j$ for $j \geq i$; and trace $r_i$ positively determines $C_i$, does not positively determine any $C_j$ for $j > i$, and does not negatively determine any $S_j$ for $j > i$.
	This is because otherwise some safety or co-safety properties either cannot be determined, which contradicts the fact that $P$ is an obligation property, or they are determined by the same finite traces, which contradicts the fact that $P$ is not a $(k-1)$-obligation property.
	Then, the worst case for $v$ is when $P$ is not $(k-1)$-obligation and it reads $r_k$ above, which forces $2k$ switches.
	One can verify that $v$ always converges to the correct property value, i.e., $\lim v(f) = \tau_P(f)$ for all $f \in \Sigma^\omega$.
	Therefore, verdict $v$ universally monitors $\tau_P$ in $\B$.
	
	For the other direction, suppose $\tau_P$ is universally monitorable by a verdict function on $\B$ that changes its value at most $2k$ times, and assume towards contradiction that $P$ is not a $k$-obligation property.
	In particular, suppose $P$ is an $m$-obligation property for some $m > k$, which cannot be expressed as a $k$-obligation, and let $P = \bigcap_{i=1}^{m} (S_i \cup C_i)$ where $S_i$ is safe and $C_i$ is co-safe for each $1 \leq i \leq m$.
	By the same argument used above, there exist finite traces $s_1 \prec r_1 \prec \ldots \prec s_{m} \prec r_{m}$, w.l.o.g., such that for every $1 \leq i \leq m$ trace $s_i$ negatively determines $S_i$, does not negatively determine any $S_j$ for $j > i$, and does not positively determine any $C_j$ for $j \geq i$; and trace $r_i$ positively determines $C_i$, does not positively determine any $C_j$ for $j > i$, and does not negatively determine any $S_j$ for $j > i$.
	Assume w.l.o.g. that $v(\varepsilon) = \tr$.
	After reading each finite trace described above, $v$ has to switch its output because otherwise we can construct a trace $f$ such that $\lim v(f) \neq \tau_P(f)$.
	But since $v$ can only change its value $2k$ times, it immediately yields that $v$ cannot universally monitor $\tau_P$ where $P$ is an $m$-obligation property for $m > k$.
	Therefore, $P$ must be a $k$-obligation property.
\end{proof}

The countable intersection of co-safety properties and the countable union of safety properties, i.e.,
so-called \emph{response} and \emph{persistence} properties \cite{Chang93}, respectively, are also universally monitorable.

\begin{theorem} \label{thm:resp}
	A boolean property $P$ is a response property iff $\tau_P$ is universally monitorable from below
	by an unrestricted verdict function on $\B_t$.
\end{theorem}

\begin{proof}
	Suppose $P$ is a response property, i.e., there exists a set $S \subseteq \Sigma^*$ such that for every $f \in \Sigma^\omega$ we have $f \in P$ iff infinitely many prefixes of $f$ belong to $S$.
	Let $v$ be a verdict function as follows: $v(s) = \tr$ if $s \in S$, and $v(s) = \fa$ if $s \notin S$.
	Now, let $f \in \Sigma^\omega$ be a trace.
	We have $\tau_P(f) = \tr$ iff for every $s \prec f$ there exists $r \in \Sigma^*$ such that $sr \prec f$ and $sr \in S$ iff $\limsup v(f) = \tr$.
	Therefore, $v$ universally monitors $\tau_P$ from below in $\B_t$.
	
	Now, suppose there exists a verdict function $v$ that universally monitors $\tau_P$ from below in $\B_t$.
	Because $v$ is a function on $\B_t = \{ \tr, \fa \}$, there is a set $S \subseteq \Sigma^*$ such that $v(s) = \tr$ for all $s \in S$, and $v(s) = \fa$ for all $s \notin S$.
	Then, we get that $\limsup v(f) = \tr$ iff for every $s \prec f$ there exists $r \in \Sigma^*$ such that $sr \prec f$ and $sr \in S$ iff $f \in P$.
	Observe that the set $S$ is exactly as in the definition of a response property, therefore $P$ is a response property.
\end{proof}

The proof for persistence properties is symmetric.

\begin{theorem} \label{thm:pers}
	A boolean property $P$ is a persistence property iff $\tau_P$ is universally monitorable
	from below by an unrestricted verdict function on $\B_f$.
\end{theorem}

Positive, finite boolean combinations of response and persistence properties are called \emph{reactivity} properties \cite{Chang93}.
We consider existential monitorability in $\B_\bot$ by unrestricted verdict functions, and provide a lower bound.

\begin{theorem} \label{thm:react}
	For every boolean reactivity property $P$,
	we have that $\tau_P$ is existentially monitorable from below by an unrestricted verdict function on $\B_\bot$.
\end{theorem}

\begin{proof}
	Suppose $P$ is a boolean reactivity property, i.e., $P = \bigcap_{i=1}^k (R_i \cup P_i)$ for some $k \geq 1$ where $R_i$ is a response and $P_i$ is a persistence property for every $1 \leq i \leq k$.
	By Theorem~\ref{thm:resp}, each $\tau_{R_i}$ is universally monitorable from below by a verdict function $u_i$ on $\B_t$.
	For each $1 \leq i \leq k$, consider the verdict function $v_i$ on $\B_\bot$ defined as follows: let $v_i(s) = \tr$ if $u_i(s) = \tr$ or $s$ positively determines $R_i$, let $v_i(s) = \fa$ if $s$ negatively determines $R_i$, and $v_i(s) = \bot$ otherwise.
	Note that each $v_i$ existentially monitors $\tau_{R_i}$ from below, and for every $f \in \Sigma^\omega$, if $f \in R_i$ for every $1 \leq i \leq k$, then $f \in P$.
	Then, we can construct the verdict $v$ to monitor $\tau_P$:
	Let $v(\varepsilon) = \tr$ and let $x$ be a memory for $v$ that initially contains $\varepsilon$.
	On non-empty traces, $v$ outputs $\bot$ until it observes a trace $s$ such that for every $1 \leq i \leq k$ there exists $r_i$ such that $x \prec r_i \preceq s$ and $v_i(r_i) = \tr$.
	When $v$ reads such a trace $s$, it outputs $\tr$, updates $x$ to store $s$, and outputs $\bot$ until the next trace that satisfies the condition above.
	Observe that, for every $f$, if $\limsup v(f) = \tr$ then $\tau_P(f) = \tr$; and for every $s$ there exists $f$ such that $\limsup v(sf) = \tau_P(sf)$ unless $R_i$ is negatively determined for some $1 \leq i \leq k$.
	
	If for some response component $R_i$ is negatively determined, then we can switch to monitor corresponding persistence component instead.
	Consider the verdict  $u_i'$ for $\tau_{P_i}$ on $\B_f$ and construct $v_i'$ on $\B_\bot$ as follows: let $v_i'(s) = \fa$ if $u_i'(s) = \fa$ or $s$ negatively determines $P_i$, let $v_i'(s) = \tr$ if $s$ positively determines $P_i$, and $v_i'(s) = \bot$ otherwise.
	One can verify that for every $f$, if $\limsup v_i'(f) = \fa$ then $\tau_P(f) = \fa$; and for every $s$ there exists $f$ such that $\limsup v_i'(sf) = \tau_P(sf)$ unless $P_i$ is positively determined.
	Once $P_i$ is positively determined, we know that all possible future traces satisfy $R_i \cup P_i$.
	Then, we can switch to the previous procedure of monitoring the response components, excluding $R_i$, and repeat as many times as necessary.
\end{proof}

We now demonstrate that Theorem~\ref{thm:react} is indeed a lower bound for the capabilities of
existential monitors in $\B_\bot$.

\begin{example} \label{ex:react}
	Let $P = \bigcup_{i \in \N} R_i$ such that each $R_i$ is a response property.
	In particular, the property $P$ belongs to the class of $G_{\delta\sigma}$ sets in the Borel hierarchy, which strictly contains the reactivity properties.
	Moreover, suppose that for some $j \in \N$ property $R_j$ is live.
	Let $u$ be a verdict on $\B_t$ for $\tau_{R_j}$.
	We construct a verdict $v$ on $\B_\bot$ for $\tau_P$ as follows: $v(s) = \tr$ if $u(s) = \tr$, and $v(s) = \bot$ otherwise.
	Clearly, for every $f \in \Sigma^\omega$, if $f \in R_j$ then $f \in P$, and thus $\limsup v(f) \leq \tau_P(f)$.
	Moreover, since $R_j$ is live, so is $P$, i.e., every finite trace $s$ can be extended with some $f$ such that $sf \in P$,
	and thus $\limsup v(sf) = \tau_P(sf)$.
	It follows that $P$ is existentially monitorable from below by $v$ on $\B_\bot$.
\end{example}

The remaining combinations of monitoring modality and value domain allow us to monitor every boolean property.

\begin{theorem} \label{thm:monoexbt}
	For every boolean property $P$,
	we have that $\tau_P$ is existentially monitorable from below by a monotonically increasing verdict function on $\B_t$.
\end{theorem}

\begin{proof}
	Let $v$ be a verdict function such that $v(s) = \tr$ if $s$ positively determines $P$, and $v(s) = \fa$ otherwise.
	Observe that $v$ is indeed monotonically increasing in $\B_t$, and $\limsup v(f) \leq \tau_P(f)$ for every $f \in \Sigma^\omega$.
	Let $s \in \Sigma^*$ be an arbitrary trace.
	If $v(s) = \tr$, then $\tau_P(sf) = \tr$ for every continuation $f \in \Sigma^\omega$; otherwise, there exists some $f \in \Sigma^\omega$ such that $\tau_P(sf) = \fa$.
	It implies that for every $s \in \Sigma^*$ there exists $f \in \Sigma^\omega$ such that $\limsup v(sf) = \tau_P(sf)$.
	Therefore, $\tau_P$ is existentially monitorable from below by $v$.
\end{proof}

Combining Proposition~\ref{pro:exmonbtbf} and Theorem~\ref{thm:monoexbt}, we carry this result over to domain $\B_f$.

\begin{corollary} \label{cor:monoexbf}
	For every boolean property $P$,
	we have that $\tau_P$ is existentially monitorable from below by a monotonically increasing verdict function on $\B_f$.
\end{corollary}

Tables~\ref{tab:univ}~and~\ref{tab:exis} summarize the results of this section.
The classes of safety, co-safety, obligation, response, persistence, reactivity, and classically monitorable boolean properties
are denoted by $\mathsf{Safe}$, $\mathsf{CoSafe}$, $\mathsf{Obl}$, $\mathsf{Resp}$, $\mathsf{Pers}$, $\mathsf{React}$,
and $\mathsf{Mon}$, respectively.
We note that the upper bound for unrestricted existential monitors on $\B_\bot$ is an open problem.

\begin{table}[ht] 
	\centering
	\caption{Correspondence between classes of boolean properties and universal monitorability. \label{tab:univ}}
	\scalebox{0.94}{
		\begin{tabular}[t]{cccc}
			\multicolumn{1}{c}{} & \multicolumn{3}{c}{Universally monitorable from below} \\ \cmidrule{2-2} \cmidrule{4-4}
			\multicolumn{1}{l}{$\D$} & Monotonically increasing && Unrestricted verdict \\ \hline
			\multicolumn{1}{l}{$\B$} & $\emptyset$ or $\Sigma^\omega$ (Rem.~\ref{rem:allornothing}) && $\mathsf{Obl}$ (Thm.~\ref{thm:obl}) \\ 
			\multicolumn{1}{l}{$\B_\bot$} & $\mathsf{Safe} \cap \mathsf{CoSafe}$ (Thm.~\ref{thm:safencosafe}) && $\mathsf{Obl}$ (Prop.~\ref{pro:equiv}~+~Thm.~\ref{thm:obl})\\ 
			\multicolumn{1}{l}{$\B_t$} & $\mathsf{CoSafe}$ (Thm.~\ref{thm:safe+cosafe}) && $\mathsf{Resp}$ (Thm.~\ref{thm:resp}) \\ 
			\multicolumn{1}{l}{$\B_f$} & $\mathsf{Safe}$ (Thm.~\ref{thm:safe+cosafe}) && $\mathsf{Pers}$ (Thm.~\ref{thm:pers}) \\
	\end{tabular}}
\end{table}

\vspace*{-0.16cm}

\begin{table}[h!] 
	\centering
	\caption{Correspondence between classes of boolean properties and existential monitorability. \label{tab:exis}}
	\scalebox{0.94}{
		\begin{tabular}[t]{cccc}
			\multicolumn{1}{c}{} & \multicolumn{3}{c}{Existentially monitorable from below} \\ \cmidrule{2-2} \cmidrule{4-4}
			\multicolumn{1}{l}{$\D$} & Monotonically increasing && Unrestricted verdict \\ \hline
			\multicolumn{1}{l}{$\B$} & $\emptyset$ or $\Sigma^\omega$ (Rem.~\ref{rem:allornothing}) && $\mathsf{Obl}$ (Prop.~\ref{pro:equiv}~+~Thm.~\ref{thm:obl}) \\ 
			\multicolumn{1}{l}{$\B_\bot$} & $\mathsf{Mon}$ (Prop..~\ref{thm:strmon}) && at least $\mathsf{Mon \cup React}$ (Thm.~\ref{thm:react}) \\
			\multicolumn{1}{l}{$\B_t$} & any $P \subseteq \Sigma^\omega$ (Thm.~\ref{thm:monoexbt}) && any $P \subseteq \Sigma^\omega$ (Thm.~\ref{thm:monoexbt}) \\
			\multicolumn{1}{l}{$\B_f$} & any $P \subseteq \Sigma^\omega$ (Cor.~\ref{cor:monoexbf}) && any $P \subseteq \Sigma^\omega$ (Cor.~\ref{cor:monoexbf}) \\ 
	\end{tabular}}
\end{table}

We conclude the section with a simple example that demonstrates the concept of precision in the context of boolean properties.

\begin{example}
	Let $P = \Diamond (a \lor b \lor c)$, where $\Diamond$ is the \emph{eventually} operator \cite{Piterman18},
	and $\tau_P$ be the corresponding quantitative property.
	Consider the following verdict functions on $\B_t$:
	\begin{itemize}
		\item $v_a(s) = \tr$ iff $s$ contains $a$,
		\item $v_{ab}(s) = \tr$ iff $s$ contains $a$ or $b$,
		\item $v_{bc}(s) = \tr$ iff $s$ contains $b$ or $c$,
		\item $v_{abc}(s) = \tr$ iff $s$ contains $a$ or $b$ or $c$.
	\end{itemize}
	All these verdict functions are monotonic.
	Moreover, functions $v_a$, $v_{ab}$, and $v_{bc}$ existentially monitor $\tau_P$ from below while $v_{abc}$ monitors universally.
	
	Observe that $v_{ab}$ is more precise than $v_a$, because for every finite prefix $s$ that yields $v_a(s) = \tr$,
	we also get $v_{ab}(s) = \tr$, but not vice versa, considering the traces that contain $b$ but not $a$.
	However, we cannot compare $v_{ab}$ and $v_{bc}$, as for every  $s \prec a^\omega$,
	we have $v_{bc}(s) \leq v_{ab}(s)$ and $\limsup v_{bc}(a^\omega) \leq \limsup v_{ab}(a^\omega)$, and vice versa for $c^\omega$.
	Finally, $v_{abc}$ is the most precise among these verdicts as it universally monitors $\tau_P$. 
\end{example}

\section{Approximate Register Monitors} \label{sect:reg}
\subsection{Verdicts generated by register machines}

In this section, we follow \cite{Ferrere20} to define \emph{register machines} as a model for generating an output stream
that represents a verdict sequence for monitoring quantitative properties.
We consider a set of integer-valued \emph{registers} denoted $X$.
A \emph{valuation} $\mathsf{v} : X \to \Z$ is a mapping from the set of registers to integers.
An \emph{update} is a function from valuations to valuations,
and a \emph{test} is a function from valuations to $\B$.
The set of updates over $X$ is denoted by $\Gamma(X)$, and the set of tests by $\Phi(X)$.
We describe updates and tests over $X$ using integer- and boolean-valued expressions, called \emph{instructions}. 

\begin{definition}
	A (deterministic) \emph{register machine} is a tuple $M\!=\!(X,Q,\Sigma,\Delta,q_0,\D,\lambda)$,
	where $X$ is a finite set of registers,
	$Q$ is a finite set of states,
	$\Sigma$ is a finite alphabet,
	$\Delta \subseteq Q \times \Sigma \times \Phi(X) \times \Gamma(X) \times Q$ is a set of edges,
	$q_0 \in Q$ is the initial state.
	$\D$ is an output value domain, and 
	$\lambda : Q \times \Z^{|X|} \to \D$ is an output function.
	Moreover, for every state $q \in Q$, letter $\sigma \in \Sigma$, and valuation $\mathsf{v}$,
	there is exactly one outgoing edge $(q,\sigma,\phi,\gamma,q') \in \Delta$ with $\mathsf{v}\models\phi$.
\end{definition}

Let $M=(X,Q,\Sigma,\Delta,q_0,\D,\lambda)$ be a register machine. 
A pair consisting of a state $q \in Q$ and a valuation $\mathsf{v} : X \to \Z$ constitute a \emph{configuration} of $M$.
The initial configuration $(q_0,\mathsf{v_0})$ of $M$ is such that $\mathsf{v_0}(x) = 0$ for every $x \in X$.
Between two configurations of $M$, the \emph{transition} relation is defined by
$(q,\mathsf{v}) \xrightarrow{\sigma} (q',\mathsf{v'})$ iff
there exists an edge $(q,\sigma,\phi,\gamma,q') \in \Delta$ such that
$\mathsf{v} \models \phi$ and $\mathsf{v'}=\gamma(\mathsf{v})$.
On an infinite word $f=\sigma_1 \sigma_2 \ldots$,
the machine $M$ produces an infinite sequence of transitions
$(q_0,\mathsf{v_0}) \xrightarrow{\sigma_1} (q_1,\mathsf{v_1}) \xrightarrow{\sigma_2} \cdots$,
and an infinite \emph{output sequence} $(\lambda(q_i,\mathsf{v_i}))_{i \in \N}$.

\begin{definition}
	A register machine $M \hspace{-.25em}=\hspace{-.25em} (X,Q,\Sigma,\Delta,q_0,\D,\lambda)$ \emph{generates} the verdict function $v : \Sigma^* \to \D$ iff,
	for every finite trace $s \in \Sigma^*$, the machine $M$ after reading $s$ reaches a configuration $(q,\mathsf{v})$ such that
	$\lambda(q,\mathsf{v}) = v(s)$.
\end{definition}

We mainly focus on a simple form of register machines which can only increment, reset, and compare registers.
The according instruction set is denoted by $\langle 0, +1, \geq \rangle$, and equivalent to the instruction set
$\langle 0, +1, -1, \geq\!0 \rangle$, as was shown in \cite{Ferrere18}.

\begin{definition}
	A \emph{counter machine} is a register machine with the instructions $x \gets 0$, $x \gets x + 1$, and $x \geq y$
	for registers $x,y \in X$, and an output function that in every state outputs 0, $\infty$, or one of the register values.
	A verdict function $v$ is a \emph{$k$-counter verdict function} iff
	$v$ is generated by a counter machine $M$ with $k$ registers.
	A quantitative property $p$ is \emph{$k$-counter monitorable} iff
	there is a $k$-counter verdict function that monitors $p$.
\end{definition}

Note that we can use the various modalities of monitoring defined in Section~\ref{sect:defn}.
For instance, a property $p$ is existentially $k$-counter monitorable from below iff
$p$ is $k$-counter monitorable and the witnessing verdict function existentially monitors $p$ from below.

One can also define \emph{extended counter machines} with generic output functions.
For example, a verdict function generated by an extended 3-counter machine (with an output function that can perform division)
can universally monitor the average response-time property from above, as demonstrated in Example~\ref{ex:art}.

We remark that our model of register machines is more general than register transducer models operating over uninterpreted infinite alphabets, which typically cannot count (see, e.g., \cite{Khalimov18}).

\subsection{Precision-resource trade-offs for register machines}

In the following example, we illustrate how the arithmetic operations of register machines can play a role in
precision-resource trade-offs for monitoring.

\begin{example} [Adders versus counters]
	Let $\Sigma=\{a,b\}$ and $\D = \N \cup \{\infty\}$.
	Let $p$ be a property such that $p(f) = 2^n$, where $n$ is the length of the longest uninterrupted sequence of
	$a$'s in $f \in \Sigma^\omega$.
	Consider a 2-register machine $M$ with the following instructions:
	$x \gets 1$,  $x \gets x + y$, and $x \geq y$ for $x,y \in X$.
	When $M$ starts reading a segment of $a$'s, it resets one of its registers, say $x$, to 1 and doubles its value
	after each $a$.
	After the segment ends, it compares the value of $x$ with the other register, say $y$, and stores the maximum in $y$,
	which determines the output value.
	This way $M$ can generate $v_{\textit{add}}$ such that $v_{\textit{add}}(s) = 2^n$,
	where $n$ is the length of the longest uninterrupted sequence of $a$'s in $s \in \Sigma^*$.
	Verdict $v_{\textit{add}}$ is monotonically increasing and it universally monitors $p$.
	Now, suppose that we have, instead, a verdict $v_{\textit{count}}$ that is generated by a 2-counter machine.
	Since the counter values can only grow linearly, we can have $v_{\textit{count}}(s) = 2n$ for $n$ as above.
	Although it grows much slower, $v_{\textit{count}}$ existentially monitors $p$ from below,
	because the extension $a^\omega$ yields $\limsup v_{\textit{count}}(s a^\omega) = p(s a^\omega) = \infty$ for every $s \in \Sigma^*$.
	Since $v_{\textit{add}}$ universally monitors $p$, it is clearly more precise than $v_{\textit{count}}$.
\end{example}

Recall the two-pair maximal response-time property from Section~\ref{sect:intro}.
We can generalize this property
to give an example for a precision-resource trade-off
on the number of counter registers that are available for monitoring.

\begin{example} [Counter machine]
	Let $k \in \N$ and let $\Sigma_k = \{\req_1,\ack_1,\ldots,\req_k,\ack_k,\other\}$.
	The $k$-pair maximal response-time property $p : \Sigma^\omega \to (\N \cup \{\infty\})^k$
	specifies the maximal response times for all $(\req,\ack)$ pairs in $\Sigma_k$.
	More explicitly, for every $1 \leq i \leq k$, let $p_i$ specify the maximal response-time for pair
	$(\req_i,\ack_i)$ as in Example~\ref{ex:mrt}, and let $p(f) = (p_1(f),\ldots,p_k(f))$ for every $f \in \Sigma^\omega$.
	As hinted in Section~\ref{sect:intro}, there is a $2k$-counter verdict function $v_{2k}$ which simply
	combines the 2-counter verdict functions $u_i$ that universally monitor $p_i$ for all $1 \leq i \leq k$. 
	Specifically, $v_{2k}(s) = (u_1(s), \ldots, u_k(s))$ for every $s \in \Sigma^*$.
	
	Observe that a $(k+1)$-counter verdict function $v_{k+1}$ cannot universally monitor $p$, because whenever it reads a trace that contains more than one
	active request, it needs to either ignore some active requests and process only one of them, or forget the maximal
	response time for some pairs and use those counters to process the active requests.
	However, it can existentially monitor $p$ from below, because every finite trace can be extended with a continuation
	$f$ in which all previously active requests are acknowledged and the true maxima occur in $f$ one by one.
	If the server has at most one active request at any given time, then $k+1$ counters suffice for universal monitoring.
	This is because it only needs to use one register to process the current response time while storing the maxima in the
	remaining $k$ counters.
	Let $p'$ be the variant of $p$ under the assumption of no simultaneous requests,
	and suppose we have a $(\frac{k}{2}+1)$-counter verdict function $v_{\frac{k}{2}+1}$.
	By the same reason that $v_{k+1}$ cannot universally monitor $p$,
	the function $v_{\frac{k}{2}+1}$ cannot universally monitor $p'$.
	However, for every odd number $1 \leq i < k$, we can assign one counter to store $\max(u_i(s), u_{i+1}(s))$,
	which provides an over-approximation for either $p_i$ or $p_{i+1}$ while being precise for the other.
	Overall, although it can provide a meaningful approximation, the function $v_{\frac{k}{2}+1}$ is less precise than $v_{k+1}$.
\end{example}

The following theorems generalize this example.

\begin{theorem} \label{thm:ctrprecise}
	For every $k > 1$, there exists a quantitative property $p_k$ such that $p_k$ is universally monitorable
	by a monotonically decreasing $k$-counter verdict function $v_k$.
	Moreover, for every $\ell < k$ and every monotonically decreasing $\ell$-counter verdict function $v_\ell$ that
	approximately monitors $p_k$ from below (resp. above),
	there exists a monotonically decreasing $(\ell+1)$-counter verdict function $v_{\ell+1}$ that approximately monitors $p_k$ from below (resp. above) such that
	$v_{\ell+1}$ is more precise than $v_\ell$.
\end{theorem}

\begin{proof}
	For convenience, we consider the $\langle 0, +1, -1, \geq\!0 \rangle$ variant of counter machines.
	Let $\Sigma_k = \{1, \ldots, k\}$.
	For every $s \in \Sigma_k^*$ and $i \in \Sigma_k$ denote by $|s|_i$ the number of occurrences of the letter $i$ in $s$.
	Consider the boolean safety property $P_k = \{f \in \Sigma_k^\omega \;|\; \forall \, 1 \leq i < k : \forall s \prec f : |s|_i \geq |s|_{i+1}\}$, and let $p_k$ be as follows:
	$p_k(f) = \infty$ if $f \in P_k$, and $p_k(f) = |r|$ otherwise, where $r$ is the shortest prefix of $f$ that negatively determines $P_k$.
	We construct a verdict function $v_k$ as follows:
	$v_k(s) = \infty$ if $s$ does not negatively determine $P_k$, and $v_k(s) = |r|$ otherwise, where $r$ is the shortest prefix of $s$ that negatively determines $P_k$. 
	The verdict $v_k$ is monotonically decreasing and it can be generated by a \\ $k$-counter machine where, for every $1 \leq i < k$ and $s \in \Sigma_k^*$, the counter $x_i$ stores $|s|_i - |s|_{i+1}$, and $x_k$ stores $|s|$. 
	Moreover, because we need $k-1$ counters to recognize $P_k$ (see Thm. 4.3 in \cite{Ferrere18})
	and one more to store the output, $p_k$ is not universally $\ell$-counter monitorable for $\ell < k$.

	Let $\ell < k$, and take a monotonically decreasing $\ell$-counter verdict function $v_\ell$ that approximately monitors $p_k$ from below.
	Note that, for every $s \in \Sigma_k^*$, if the generating counter machine does not store a linear function $\alpha(s) \leq |s|$, then $v_\ell$ can be either the constant 0 function or a function that switches from $\infty$ to 0 and never misses a violation.
	Then, we construct an $(\ell+1)$-counter machine that stores $|s|$ in the new counter.
	If $v_\ell$ is constant, it uses the rest to count $|s|_i - |s|_{i+1}$ for $1 \leq i < \ell$ and catch violations, similarly as $v_k$ above; otherwise, it outputs $|s|$ instead of 0.
	The resulting verdict $v_{\ell+1}$ is monotonically decreasing and approximately monitors $p_k$ from below.
	It is also more precise than $v_\ell$.
	Now, suppose the generating $\ell$-counter machine counts a linear function $\alpha(s) \leq |s|$.
	Since this machine cannot recognize $P_k$, there exists $f \in P_k$ such that $\lim v_\ell(f) < \infty$, i.e., $v_\ell$ incorrectly concludes that $|s|_i < |s|_{i+1}$ for some $s \prec f$ and $1 \leq i < k$.
	We construct an $(\ell+1)$-counter machine $M$ where the additional counter keeps track of $|s|_i - |s|_{i+1}$ for every $s \in \Sigma_k^*$, the output register stores $|s|$, and the rest behave the same as in $v_\ell$.
	Moreover, whenever $v_\ell$ concludes that $|s|_i < |s|_{i+1}$, the behavior of $M$ is determined by
	the correct value of $|s|_i - |s|_{i+1}$ stored in the new counter.	
	It yields that the verdict $v_{\ell+1}$ generated by $M$ is more precise than $v_\ell$ because $\lim v_\ell(f)  < \lim v_{\ell+1}(f)$ for some trace $f \in \Sigma^\omega$.
	The case for monitoring from above is similar.
\end{proof}

Since monitoring $p_k$ in the proof above involves recognizing a boolean property $P_k$, the counter machine for $p_k$ must be able to distinguish traces with respect to $P_k$.
To achieve this, intuitively, the machine needs a counter for each ``independent'' quantity.
Also, the use of a variable-size alphabet $\Sigma_k$ is merely a convenience.
We can encode every word over $\Sigma_k$ in binary with the help of an additional separator symbol.
More explicitly, we can take a ternary alphabet $\Sigma = \{0,1,\#\}$ to represent every letter in $\Sigma_k$ as
a binary sequence and separate the sequences by $\#$.
We combine these observations to construct a quantitative property for which counting does not suffice for
universal monitoring no matter the number of registers, but each additional register gives a better approximation.

\begin{theorem}
	There exists a quantitative property $p$ such that for every $k > 1$ and every $k$-counter verdict function $v_k$
	that approximately monitors $p$ from below (resp. above),
	there exists a $(k+1)$-counter verdict function $v_{k+1}$ that approximately monitors $p$ from below (resp. above) and
	is more precise than $v_k$.
\end{theorem}

\begin{proof}
	Let $\Sigma = \{0,1,\#\}$.
	For every $s \in \Sigma^*$ and $i \in \N$, let $n_i(s)$ denote the number of occurrences of the binary sequence that corresponds to $i$ in the longest prefix of $s$ that ends with $\#$.
	For example, if we have $s = 001\#10$, then $n_1(s)=1$ and $n_2(s)=0$.
	Similarly as in the proof of Theorem~\ref{thm:ctrprecise}, consider counter machines with instructions $\langle 0, +1, -1, \geq\!0 \rangle$ and the following boolean safety property:
	$P \!=\! \{f \in \Sigma^\omega \,|\, \forall i \in \N : \forall s \prec f : n_i(s) \geq n_{i+1}(s)\}$.
	We define the quantitative property $p$ as $p(f) = \infty$ if $f \in P$, and $p(f) = |r|_{\#}$ where $r$ is the shortest prefix of $f$ that negatively determines $P$.
	Observe that $P$ is a generalization of $P_k$ in the proof of Theorem~\ref{thm:ctrprecise}, and it requires counting infinitely many distinct quantities.
	Therefore, one can show that, to universally monitor $p$, one needs infinitely many counter registers.
	However, for every $k > 1$, there exists a $k$-counter verdict function $v_k$ that approximately monitors $p$ from below or above, for instance, by keeping track of $n_i(s) - n_{i+1}(s)$ for every $1 \leq i < k$ and counting $\#$'s in the remaining register.

	We now construct a $(k+1)$-counter verdict function $v_{k+1}$ from $v_k$.
	Suppose $v_k$ approximately monitors $p$ from below.
	Note that the generating machine of $v_k$ lacks the resources to distinguish traces with respect to $P$ correctly.
	Therefore, regardless of the monotonicity of $v_k$, a similar reasoning as in the proof of Theorem~\ref{thm:ctrprecise} applies.
	We can use the additional counter of $v_{k+1}$ to keep track of $n_i(s) - n_{i+1}(s)$ or $|s|_\#$ for every $s \in \Sigma^*$ while the rest operate the same as in $v_k$.
	It yields that $\limsup v_k (f) < \limsup v_{k+1}(f)$; thus $v_{k+1}$ is more precise than $v_k$.
	One can similarly show the case for monitoring from above.
\end{proof}

\section{Conclusion and Future Work} \label{sect:conc}

We argued for the need of a quantitative semantic framework for runtime verification which supports
monitors that over- or under-approximate quantitative properties,
and we provided such a framework.

An obvious direction for future work is to systematically explore precision-resource tradeoffs
for different monitor models and property classes.
For example, a quantitative property class that we have not considered in this work is the limit monitoring
of statistical indicators \cite{Ferrere20}.
Other interesting resources that play a role in precision-resource trade-offs are the ``speed''
or rate of convergence of monitors, that is, how quickly a monitor reaches the desired property value,
and ``assumptions'', that is, prior knowledge about the system or the environment that can be used by the
monitor \cite{Henzinger20, Aceto21}.
We also plan to consider the reliability of communication channels \cite{Kauffman20} and how it relates to
monitoring precision.

Another question is the synthesis problem: given a quantitative property $p$ and a register machine template (instruction set and number of registers), does there exist a register machine $M$ generating a verdict function $v$ that universally or existentially monitors $p$ from below or above?

Building on our definitions of continuous and co-continuous quantitative properties, one can define a
generalization of the safety-progress hierarchy \cite{Chang93} to obtain a Borel classification of quantitative properties. 

Lastly, a logical extension of monitoring is \emph{enforcement} \cite{Ligatti10, Falcone11b, Falcone12}, that is,
manipulating the observed system's behavior to prevent undesired outcomes.
We aim to extend the notion of enforceability from the boolean to the quantitative setting and explore precision-resource
trade-offs for enforcement monitors (a.k.a. \emph{shields} \cite{Konighofer17}).

\section*{Acknowledgment}
We thank the anonymous reviewers for their helpful comments.
This research was supported in part by the Austrian Science Fund (FWF) under grant Z211-N23 (Wittgenstein Award).

\balance
\bibliographystyle{IEEEtran}
\bibliography{qam}

\end{document}